\documentclass[11pt,a4paper]{article}
\usepackage[margin=1in]{geometry}
\usepackage{amsmath,amssymb,amsfonts,setspace,url}
\usepackage{bm}
\usepackage{latexsym,graphicx}
\usepackage{amsthm}
\usepackage[singlelinecheck=false]{caption}
\usepackage{subfig}
\usepackage{float}
\usepackage{algorithmic}
\usepackage[ruled]{algorithm}
\usepackage{bigstrut,array,multirow}
\usepackage{here}
\usepackage{color}
\usepackage{hyperref}
\usepackage[capitalise]{cleveref}
\usepackage{mathpazo}

\theoremstyle{plain}
\newtheorem{theorem}{Theorem}

\newtheorem{lemma}{Lemma}

\newtheorem{corollary}{Corollary}

\theoremstyle{definition}

\providecommand{\w}{w}
\newcommand{\myw}{\mathcal{W}}

\newcommand{\abs}[1]{| #1|}

\newcommand{\overbar}[1]{\mkern 1.5mu\overline{\mkern-1.5mu#1\mkern-1.5mu}\mkern 1.5mu}

\newcommand{\floor}[1]{\left\lfloor #1 \right\rfloor}

\providecommand{\Omegae}{\Omega}

\providecommand{\MAXSAT}{\mathrm{MAXSAT}}
\providecommand{\kSAT}{k\textrm{-}\mathrm{SAT}}
\providecommand{\MAXkSAT}{\mathrm{MAX}\textrm{-}k\textrm{-}\mathrm{SAT}}
\providecommand{\MAXEkSAT}{\mathrm{MAX}\textrm{-}\mathrm{E}k\textrm{-}\mathrm{SAT}}
\providecommand{\MAXCSP}{\mathrm{MAXCSP}}

\newcommand{\poly}{\mathrm{poly}}

\newcommand{\mynb}{\Cc}

\providecommand{\Cc}{\mathcal{C}}
\providecommand{\Dd}{\mathcal{D}}

\begin{document}

\title{A Simpler Exponential-Time Approximation Algorithm for $\MAXkSAT$}
\date{}
    \author{Harry Buhrman\thanks{Quantinuum and QuSoft.}
    \and 
    Sevag Gharibian\thanks{Department of Computer Science and Institute for Photonic Quantum Systems, Paderborn University.}
    \and
    Zeph Landau\thanks{Department of Computer Science, University of California, Berkeley.}
    \and
    Fran{\c c}ois Le Gall\thanks{Graduate School of Mathematics, Nagoya University.}
    \and
    Norbert Schuch\thanks{Faculty of Mathematics and Faculty of Physics, University of Vienna.}
    \and
    Suguru Tamaki\thanks{Graduate School of Information Science, University of Hyogo.}
 }   
% \author{Harry Buhrman\\
% Quantinuum \& QuSoft
% \and
% Sevag Gharibian\\
% Paderborn University
% \and
% Zeph Landau\\
% University of California, Berkeley
% \and
% Fran{\c c}ois Le Gall\\
% Nagoya University
% \and
% Norbert Schuch\\
% University of Vienna
% \and
% Suguru Tamaki\\
% University of Hyogo}

\maketitle
\thispagestyle{empty}

\begin{abstract}
We present an extremely simple polynomial-space exponential-time $(1-\varepsilon)$-approxima{-}tion algorithm for $\MAXkSAT$ that is (slightly) faster than the previous known polynomial-space $(1-\varepsilon)$-approximation algorithms by Hirsch (Discrete Applied Mathematics, 2003) and Escoffier, Paschos and Tourniaire (Theoretical Computer Science, 2014). Our algorithm repeatedly samples an assignment uniformly at random until finding an assignment that satisfies a large enough fraction of clauses. Surprisingly, we can show the efficiency of this simpler approach by proving that in any instance of $\MAXkSAT$ (or more generally any instance of $\MAXCSP$), an exponential number of assignments satisfy a fraction of clauses close to the optimal value.
\end{abstract}

\section{Introduction}\label{sec:intro}
\paragraph{Background.}
In 1997, Paturi, Pudl{\'{a}}k and Zane~\cite{PaturiPZ97} gave a simple randomized algorithm for $\kSAT$ running in time ``moderately'' exponential in the number of variables $n$, i.e., in time 
\begin{equation}\label{eq:mod}
    O(2^{(1-c/k)n}) \,\,\,\textrm{for some universal constant $c>0$}.
\end{equation}
While the constant $c$ has been improved over the years (e.g., by Paturi, Pudl{\'{a}}k, Saks and Zane~\cite{PaturiPSZ98} and Sch\"oning~\cite{SchoningFOCS99}), no similar running time is known for $\MAXkSAT$ with $k\ge 3$.\footnote{For $k=2$, Williams \cite{WilliamsICALP04} showed how to obtain running time $O^\ast(2^{\omega n/3})$, where $\omega$ is the exponent of matrix multiplication. Here, the notation $O^\ast(\cdot)$ removes $\poly(n)$ factors. Similarly, for lower bounds the notation $\Omega^\ast(\cdot)$ will remove $\poly(n)$ factors.}
There nevertheless exist nontrivial moderately-exponential-time \emph{approximation} algorithms for $\MAXkSAT$. 
%\stnote{We might want to clarify that we are concerned with \emph{exponential time approximation schemes} as e.g. \cite{BentertF0025}, instead of polynomial time approximation algorithms.}
In this work, we focus on approximation algorithms with complexity depending on the number of variables, i.e., algorithms with complexity similar to \cref{eq:mod}. There are three main prior works \cite{Alman+SODA20,Escoffier+14,Hirsch03}, which we review below. 

\paragraph{Previous approximation algorithms.} Let $f$ be any instance of $\kSAT$ with $n$ variables and $m$ clauses. For an assignment $z\in\{0,1\}^n$, we denote by $\mynb(f,z)$ the number of clauses in $f$ satisfied by $z$, and write 
\[
m^\ast=\max_{z\in\{0,1\}^n}\mynb(f,z).
\]
The following algorithms from \cite{Alman+SODA20,Escoffier+14,Hirsch03} find 
an assignment $z$ such that
\[
\mynb(f,z)\ge (1-\varepsilon)m^\ast \,.
\]

In 2003, Hirsch \cite{Hirsch03} presented two randomized polynomial-space algorithms for this task.
Hirsch's first algorithm (\cite[Theorem~1]{Hirsch03}) has complexity
\begin{equation}\label{eq:h}
O^\ast\left(\left(2-\frac{2\varepsilon}{\varepsilon+k+\varepsilon k}\right)^n\right)= O^\ast\left(2^{\left(1+\log_2\left(1-\frac{\varepsilon}{\varepsilon+k+\varepsilon k}\right)\right)n}\right)\,.
\end{equation}
It selects a random assignment and tries to improve it by repeatedly flipping a value of a variable chosen randomly from an unsatisfied clause.  
Hirsch's second algorithm (\cite[Theorem~4]{Hirsch03}) combines this idea with Sch\"oning's algorithm \cite{SchoningFOCS99} to achieve the following slightly better bound:
\begin{equation}\label{eq:hh}
O^\ast\left(\left(2-\frac{2\varepsilon}{k(1+\varepsilon)}\right)^n\right)= O^\ast\left(2^{\left(1+\log_2\left(1-\frac{\varepsilon}{k(1+\varepsilon)}\right)\right)n}\right)\,.
\end{equation}

In 2014, Escoffier, Paschos and Tourniaire \cite[Theorem 3]{Escoffier+14}
achieved complexity 
\begin{equation}\label{eq:EPT}
	O^\ast\left(2^{\left(1-\frac{\varepsilon}{1-\alpha}\right)n}\right)
	=
	O^\ast\left((2-\Omegae(\varepsilon))^n\right),
\end{equation}
where $\alpha$ is the best polynomial-time approximation ratio for $\MAXSAT$ (currently, $\alpha=0.796$~\cite{Avidor+05}). 
%Hereafter, $\Omegae(\cdot)$ means a lower bound that holds for small enough $\varepsilon$.
The algorithm from \cite{Escoffier+14} is deterministic and uses polynomial space. Note that its complexity is independent of $k$. 

In 2020, 
Alman, Chan and Williams~\cite{Alman+SODA20} obtained complexity 
\[
	O^\ast\left((2-\Omegae(\sqrt{\varepsilon}/\log(1/\varepsilon)))^n\right)
\]
with a deterministic algorithm (\cite[Theorem 3.6]{Alman+SODA20}) and complexity 
\[
	O^\ast\left((2-\Omegae(\varepsilon^{1/3}/\log^{2/3}(1/\varepsilon)))^n\right)
\]
with a randomized algorithm (\cite[discussion after Theorem 5.2]{Alman+SODA20}).
Both algorithms are more involved (for instance, they make use of expander walks and Chebyshev polynomials) and use exponential space. 

\paragraph{Our results.}
We show the following theorem. 
%\snote{I assume everything in the paper is in the unweighted case?  }
\begin{theorem}\label{th:main}
	There exists a polynomial-space randomized algorithm for $\MAXkSAT$ running in
	\[
	O^\ast\left(\left(2-\Omegae\left(\frac{\varepsilon}{k}\log\left(\frac{k}{\varepsilon}\right)\right)\right)^n\right)
	\]
 time that finds with high probability an assignment $z\in\{0,1\}^n$ such that $\mynb(f,z)\ge (1-\varepsilon)m^\ast$.
\end{theorem}
 Note that asymptotically, the bound of \cref{th:main} is better than Hirsch's algorithm (\cref{eq:hh}), and for $k=O(1)$ is also better than Escoffier-Paschos-Tourniaire's algorithm (\cref{eq:EPT}). While not faster than the algorithms by Alman, Chan and Williams \cite{Alman+SODA20} discussed above, our algorithm uses only polynomial space. To our knowledge, this is the fastest polynomial-space algorithm for $\MAXkSAT$ when $k=O(1)$.

Our algorithm is extremely simple: It samples an assignment uniformly at random and evaluates the number of clauses it satisfies. After repeating this process a given number of times, it outputs the assignment that satisfies the maximum number of clauses. Our algorithm does not work only for $\MAXkSAT$, but also for arbitrary weighted $\MAXCSP$ (\Cref{th:class2}). We prove its correctness and analyze its complexity by showing a lower bound on the number of assignments that satisfy a large number of constraints (\cref{th:lb}). %We believe that these contributions offer new insights that simplify our understanding of those fundamental computational problems.

\paragraph{Relation with prior work.}
A recent work by Buhrman, Gharibian, Landau, Le Gall, Schuch and Tamaki \cite{Buhrman+25} has considered a similar approach for computing approximate solutions of optimization problems in the context of quantum computing. Ref.~\cite{Buhrman+25} focuses on (classical and quantum) algorithms  for the $k$-local Hamiltonian problem, a quantum generalization of $\MAXkSAT$. While these algorithms can be applied to $\MAXkSAT$ (which is a special case of the $k$-local Hamiltonian problem), the running time obtained is worse than the running time of Hirsch's algorithm. In this work, we obtain a better complexity by exploiting specific properties of $\MAXkSAT$. 

%=============================================================================
\section{General Result for \texorpdfstring{$\MAXCSP$}{MAXCSP}}\label{sec:CSP}
%=============================================================================
In this section, we consider the (weighted) maximum constraint satisfaction problem ($\MAXCSP$); \Cref{th:main} for MAX-$k$-SAT follows as a corollary, as we show in \Cref{sec:SAT}. 

Let $f$ be an instance of $\MAXCSP$ with~$n$ Boolean variables $x_1,\ldots, x_n$ and $m$ constraints $C_1,\ldots,C_m$. Each constraint $C_i$ has a weight $\w_i$, which is a positive real number. We denote by~$a_i$ the number of variables in $C_i$.  The \emph{total weight} of $f$, which we denote by~$\w$, is the sum of the weights of all constraints:
\[
	\w = \sum_{i=1}^m w_i\,.
\]
We define the \emph{weighted length} of $f$, denoted by $\ell$, as 
\[
	\ell = \sum_{i=1}^m a_iw_i\,.
\]
For each variable $x_i$, we define its \emph{contribution to $f$}, denoted $\ell_i$, as the sum of the weights of all constraints including it or its negation. We obviously have 
\begin{equation}\label{eq:cont}
	\ell=\sum_{i=1}^n \ell_i \,.
\end{equation}

For an assignment $z\in\{0,1\}^n$, the \emph{weight} of $z$, which we denote by $\myw(f,z)$, is the sum of the weights of all constraints in $f$ satisfied by $z$.
We write 
\[
w^\ast=\max_{z\in\{0,1\}^n}\myw(f,z).
\]
For any $\varepsilon\in (0,1]$, let $\Dd(f,\varepsilon)$ denote the number of assignments $z\in\{0,1\}^n$ such that 
\[
\myw(f,z)\ge w^\ast-\varepsilon  w\,.
\]
(Note that the right-hand side is not $(1-\varepsilon)w^\ast$.)
Our main technical contribution is the following theorem. 
\begin{theorem}\label{th:lb}
	For any $\varepsilon\in(0,1]$,
	\[
		\Dd(f,\varepsilon)=
		\Omega^\ast\left(	
		\max_{\delta\ge 1+\frac{\varepsilon w}{\ell}}\left\{
		2^{H\left(\frac{\varepsilon w}{(\delta-1)\ell}\right) \frac{\delta-1}{\delta}n}
		\right\}
		\right).
	 \]
	Here $H(\cdot)$ denotes the binary entropy function $H(p)=-p\log p - (1-p)\log (1-p)$ for $p\in [0,1]$.\footnote{Here, we use the standard convention $0\log(0) =0$ that ensures continuity at $p=0$ and $p=1$.}
\end{theorem}
We will use the following lemma to prove \cref{th:lb}.
\begin{lemma}\label{ref:lemma}
    For any $y>0$, any $x\ge y$ and any $r\in[0,y]$,
    \[
    H\left(\frac{r}{x}\right)x\ge H\left(\frac{r}{y}\right)y\,.
    \]
\end{lemma}
\begin{proof}
The statement is trivial if $r=0$. We assume below that $r>0$.
The function $f(\alpha)=\frac{H(\alpha)}{\alpha}$ is decreasing over $(0,1]$, i.e., $f(\alpha_1)\ge f(\alpha_2)$ for any $0<\alpha_1\le \alpha_2\le 1$. By taking 
$\alpha_1=r/x$ and $\alpha_2=r/y$, we obtain
\[
    H\left(\frac{r}{x}\right)\frac{x}{r}\ge H\left(\frac{r}{y}\right)\frac{y}{r}\,, 
\]
which gives the statement.
\end{proof}
\begin{proof}[Proof of \cref{th:lb}]
For any 
\begin{equation}\label{eq:delta}
    \delta\ge 1+\frac{\varepsilon w}{\ell},
\end{equation}
let $S_\delta$ denote the set of variables that have contribution to $f$ at most $\delta \ell/n$. Since the weighted length of the formula is $\ell$, using \cref{eq:cont} we get the inequality 
\[
\left(n-\abs{S_\delta}\right)\cdot \frac{\delta \ell}{n} \le\ell,
\]
which implies
\begin{equation}\label{eq:cl}
	\abs{S_\delta}\ge \frac{(\delta-1) n}{\delta}.
\end{equation}

We now take 
\[
r=\floor{n \frac{\varepsilon w}{\delta\ell}}\,.
\] 
Observe that 
\[
r\le \frac{(\delta-1) n}{\delta}\le\abs{S_\delta}, 
\]
from \cref{eq:delta} and \cref{eq:cl}.

Let $z$ be an assignment with maximum weight, i.e., $\myw(f,z)=w^\ast$. 
Let $\Sigma\subseteq\{0,1\}^n$ denote the set of all assignments that are obtained from $z^\ast$ by flipping at most $r$ variables chosen from~$S_\delta$. 
For any $z\in \Sigma$, we have
\[
	\myw(f,z)
	\ge w^\ast - \frac{\delta \ell}{n}r
	\ge w^\ast -\varepsilon w.
\]

Observe that\footnote{Note that the first inequality is tight (up to a polynomial factor) only when $r \leq \abs{S_\delta}/2$, that is,  when $\delta \ge 1+2\frac{\varepsilon w}{\ell}$. While the lower bound could be improved for smaller $\delta$, this is enough for our purpose since the optimal value of~$\delta$ will satisfy the inequality $\delta \ge 1+2\frac{\varepsilon w}{\ell}$.} 
\begin{equation}\label{eq:tight}
\abs{\Sigma}=
\sum_{i=0}^r
\binom{\abs{S_\delta}}{i}
\ge \binom{\abs{S_\delta}}{r} \ge  \frac{2^{H\left(\frac{r}{\abs{S_\delta}}\right) \abs{S_\delta}}}{\abs{S_\delta}+1}
=\Omega^\ast\left(2^{H\left(\frac{r}{\abs{S_\delta}}\right) \abs{S_\delta}}\right)=
\Omega^\ast\left(2^{H\left(\frac{r\delta}{(\delta-1) n}\right) \frac{(\delta-1)n}{\delta}}\right),
\end{equation}
where we used \cref{ref:lemma} with $x=|S_\delta|$ and $y=\frac{(\delta-1)n}{\delta}$.
We thus obtain
\[
\abs{\Sigma}=\Omega^\ast\left(2^{H\left(\frac{\varepsilon w}{(\delta-1)\ell}\right) \frac{\delta-1}{\delta}n}\right).
\]
Optimizing over all $\delta$ satisfying \cref{eq:delta} gives the statement of the theorem.
\end{proof}

Theorem~\ref{th:lb} immediately leads to a randomized polynomial-space approximation algorithm for $\MAXCSP$. The algorithm samples an assignment uniformly at random and computes its weight by evaluating the sum of the weights of all constraints it satisfies. After repeating this process a given number of times, we output the assignment with maximum weight. Theorem~\ref{th:lb} enables us to analyze the number of repetitions needed to obtain a good approximation with high probability:  
\begin{theorem}\label{th:class2}
	For any $\varepsilon\in(0,1]$, there exists a polynomial-space randomized algorithm for $\MAXCSP$ running in
	\[
	O^\ast\left(
    \min_{\delta\ge 1+\frac{\varepsilon w}{\ell}}\left\{
	2^{\left(1-H\left(\frac{\varepsilon w}{(\delta-1)\ell}\right) \frac{\delta-1}{\delta}\right)n}\right\}
	\right)
	\]
	time that finds with high probability an assignment $z\in\{0,1\}^n$ such that $\myw(f,z)\ge w^\ast-\varepsilon w$.
\end{theorem}
While the optimal value of $\delta$ cannot be expressed analytically, when $\varepsilon$, $w$ and $\ell$ are fixed it can be found numerically. If we are only interested in the asymptotic dependence in~$\varepsilon$ of the exponent, we can simply take $\delta=2$ (as we will do in \cref{sec:SAT} to prove \cref{th:main}). 

If we know a lower bound on $w^\ast$, then we can guarantee that the assignment found by the algorithm satisfies $\myw(f,z)\ge (1-\varepsilon)w^\ast$ by decreasing the value of $\varepsilon$:
\begin{corollary}\label{th:CSP}
	Assume that we know a lower bound $\overbar{w}$ on $w^\ast$. For any $\varepsilon\in(0,1]$,
	there exists a polynomial-space randomized algorithm for $\MAXCSP$ running in
	\[
	O^\ast\left(\min_{\delta\ge 1+\frac{\varepsilon w}{\ell}}\left\{
	2^{\left(1-H\left(\frac{\varepsilon \overbar{w}}{(\delta-1)\ell}\right) \frac{\delta-1}{\delta}\right)n}\right\}
	\right)
	\]
	time that finds with high probability an assignment $z\in\{0,1\}^n$ such that $\myw(f,z)\ge (1-\varepsilon)w^\ast$.
\end{corollary}
\begin{proof}
	We apply \cref{th:lb} with $\frac{\varepsilon\overbar w}{w}$ instead of $\varepsilon$. The algorithm finds with high probability an assignment $z\in\{0,1\}^n$ such that 
	\[
		\myw(f,z)\ge w^\ast-\frac{\varepsilon\overbar w}{w} w=w^\ast-\varepsilon\overbar{w}\ge (1-\varepsilon)w^\ast,
	\]
	as claimed.
\end{proof}

%=============================================================================
\section{Application to \texorpdfstring{$\MAXkSAT$}{MAX-k-SAT}}\label{sec:SAT}
%=============================================================================
In this section, we apply the results from \cref{sec:CSP} to (unweighted) $\MAXkSAT$ and prove \cref{th:main}. We use the notations from \cref{sec:intro}. Since each clause corresponds to a constraint of unit weight, we have $w=m$ and $w^\ast=m^\ast$.

As an immediate application of \cref{th:CSP}, we obtain the following result for $\MAXEkSAT$, the version of (unweighted) $\MAXkSAT$ in which each clause has \emph{exactly} $k$ literals.
\begin{corollary}\label{cor:maxesat}
        For any $\varepsilon\in(0,1]$, there exists a polynomial-space randomized algorithm for $\MAXEkSAT$ running in
	\[
	O^\ast\left(	\min_{\delta\ge 1+\frac{\varepsilon}{k}}\left\{
	2^{\left(1-H\left(\frac{\varepsilon}{k} \frac{2^{k}-1}{2^k(\delta-1)}\right) \frac{\delta-1}{\delta}\right)n}\right\}
	\right)
	\]
	time
	that finds with high probability an assignment $z\in\{0,1\}^n$ such that $\mynb(f,z)\ge (1-\varepsilon)m^\ast$.
\end{corollary}
\begin{proof}
	For $\MAXEkSAT$ we have $\ell = k m$ and can use the trivial lower bound $\overbar w =\frac{2^{k}-1}{2^k}m$ on the maximum number of satisfied clauses.
\end{proof}
In \cref{table}, we compare the complexity of \cref{cor:maxesat} (for the optimal value of $\delta$, which is obtained numerically) with the running time of the complexity of Hirsch's algorithm (\cref{eq:hh}).

\begin{table}[h!t]
\centering
\caption{Comparison of the exponents obtained for $\MAXEkSAT$ by Hirsch's algorithm (\cref{eq:hh}) and our algorithm (\cref{cor:maxesat}) for $k\in\{3,4,5,6\}$ and several values
$\varepsilon\le 1-\frac{2^{k}-1}{2^k}$. More precisely, the table compares 
\[
1+\log_2\left(1-\frac{\varepsilon}{k(1+\varepsilon)}\right)
\]
(for Hirsch's algorithm) with 
\[
\min_{\delta\ge 1+\frac{\varepsilon}{k}}\left\{
1-H\left(\frac{\varepsilon}{k} \frac{2^{k}-1}{2^k(\delta-1)}\right) \frac{\delta-1}{\delta}\right\}
\]
(for our algorithm).
}\label{table}
\vspace{10mm}
\renewcommand{\arraystretch}{1.1}
 \begin{tabular}{|l|l|c|c|}
 	\cline{3-4}
	\multicolumn{2}{l|}{} & Hirsch's algorithm & Our algorithm \\
    \multicolumn{2}{l|}{}&(\cref*{eq:hh})&(\cref*{cor:maxesat})\\\hline
\multirow{9}{*}{$\,\,k=3\,\,$}&$\varepsilon=1/8$& 0.9455522 &0.8740555\\\cline{2-4}
&$\varepsilon=0.1$& 0.9556059 &0.8923639\\\cline{2-4}
&$\varepsilon=0.05$& 0.9769164 &0.9351926\\\cline{2-4}
&$\varepsilon=0.04$& 0.9813843 &0.9452549\\\cline{2-4}
&$\varepsilon=0.03$& 0.9859248 &0.9561051\\\cline{2-4}
&$\varepsilon=0.02$& 0.9905397 &0.9680331\\\cline{2-4}
&$\varepsilon=0.01$& 0.9952308 &0.9816589\\\cline{2-4}
&$\varepsilon=0.001$& 0.9995195 &0.9973496\\\cline{2-4}
&$\varepsilon=0.0001$& 0.9999519 &0.9996498\\\hline
\multirow{8}{*}{$\,\,k=4\,\,$}&$\varepsilon=1/16$& 0.9786263 &0.9349755\\\cline{2-4}
&$\varepsilon=0.05$& 0.9827220 &0.9450690\\\cline{2-4}
&$\varepsilon=0.04$& 0.9860608 &0.9537019\\\cline{2-4}
&$\varepsilon=0.03$& 0.9894565 &0.9629761\\\cline{2-4}
&$\varepsilon=0.02$& 0.9929106 &0.9731266\\\cline{2-4}
&$\varepsilon=0.01$& 0.9964245 &0.9846550\\\cline{2-4}
&$\varepsilon=0.001$& 0.9996396 &0.9978062\\\cline{2-4}
&$\varepsilon=0.0001$& 0.9999639 &0.9997120\\\hline
\multirow{6}{*}{$\,\,k=5\,\,$}&$\varepsilon=1/32$& 0.9912298 &0.9670797\\\cline{2-4}
&$\varepsilon=0.03$& 0.9915714 &0.9681233\\\cline{2-4}
&$\varepsilon=0.02$& 0.9943312 &0.9769253\\\cline{2-4}
&$\varepsilon=0.01$& 0.9971403 &0.9868757\\\cline{2-4}
&$\varepsilon=0.001$& 0.9997117 &0.9981403\\\cline{2-4}
&$\varepsilon=0.0001$& 0.9999711 &0.9997571\\\hline
\multirow{4}{*}{$\,\,k=6\,\,$}&$\varepsilon=1/64$& 0.9962960 &0.9834889\\\cline{2-4}
&$\varepsilon=0.01$& 0.9976173 &0.9885602\\\cline{2-4}
&$\varepsilon=0.001$& 0.9997598 &0.9983910\\\cline{2-4}
&$\varepsilon=0.0001$& 0.9999760 &0.9997908\\\hline
\end{tabular}
\end{table}

For (unweighted) $\MAXkSAT$, we obtain the best complexity by considering the number of clauses of each length. For any $i\in\{1,\ldots,k\}$, let $m_i$ denote the number of clauses of length~$i$ (note that $m=m_1+\cdots+m_k$). We can then use the lower bound
\[
	m^\ast\ge \sum_{i=1}^k \frac{2^i-1}{2^i} m_i
\]
in \cref{th:CSP}. For simplicity, we will use the weaker generic bound
\[
	m^\ast\ge \frac{m}{2}
\]
and choose $\delta=2$ in \cref{th:CSP} to derive the following result.
\begin{corollary}\label{cor:maxsat}
     For any $\varepsilon\in(0,1]$,
	there exists a polynomial-space randomized algorithm for $\MAXkSAT$ running in
	\[
	O^\ast\left(	
	2^{\left(1-\frac{1}{2}H\left(\frac{\varepsilon}{2k} \right) \right)n}
	\right)
	\]
	time
	that finds with high probability an assignment $z\in\{0,1\}^n$ such that $\mynb(f,z)\ge (1-\varepsilon)m^\ast$.
\end{corollary}
Since
\begin{align*}
    2^{\left(1-\frac{1}{2}H\left(\frac{\varepsilon}{2k}\right)\right)n}
    &=
    \left(2-\Omegae\left(H\left(\frac{\varepsilon}{2k}\right)\right)\right)^n\\
    &=
    \left(2-\Omegae\left(\frac{\varepsilon}{k}\log\left(\frac{k}{\varepsilon}\right)\right)\right)^n,
\end{align*}
\cref{cor:maxsat} implies \cref{th:main}.

\section*{Acknowledgments}%
Part of this work was done when visiting the Simons Institute for the Theory of Computing. 
SG is supported by DFG grants 563388236 and 450041824, BMFTR project PhoQuant (13N16103), and project PhoQC from the State of Northrhine Westphalia. 
ZL is supported by the U.S. Department of Energy, Office of Science, National Quantum Information Science Research Centers, Quantum Systems Accelerator, and by NSF Grant CCF 2311733. 
FLG is supported by JSPS KAKENHI grants JP20H05966, 24H00071 and MEXT Q-LEAP grant JPMXS0120319794. 
NS is supported by the Austrian Science Fund FWF (Grant DOIs 10.55776/COE1, 10.55776/P36305 and 10.55776/F71),  the European Union -- NextGenerationEU, and the European Union’s Horizon 2020 research and innovation programme through Grant No.\ 863476 (ERC-CoG \mbox{SEQUAM}).  
ST is  supported by JSPS KAKENHI grants JP20H05961, JP20H05967, JP22K11909.
%=====================================
%\bibliographystyle{plain}
%\bibliography{refs}
%\printbibliography

%=====================================
\end{document}